\documentclass{amsart}

\newtheorem{lemma}{Lemma}
\newtheorem{theorem}{Theorem}

\def\dbar{\mathrm{d}\mkern-6mu\mathchar'26} 

\def\bp{\mathbf{p}}
\def\bq{\mathbf{q}}
\def\bx{\mathbf{x}}
\def\by{\mathbf{y}}
\def\bz{\mathbf{z}}
\def\bxi{\boldsymbol{\xi}}

\def\cE{\mathcal{E}}
\def\cF{\mathcal{F}}

\newcommand{\nz}{\mathbb{N}} 
\newcommand{\cz}{\mathbb{C}} 
\newcommand{\rz}{\mathbb{R}} 

\def\gH{\mathfrak{H}}
\def\gQ{\mathfrak{Q}}
\def\gS{\mathfrak{S}}
\def\gh{\mathfrak{h}}

\def\rd{\mathrm{d}}
\def\ri{\mathrm{i}}

\def\p{{\mathbf{\hat p}}} 
\def\tr{\mathop\mathrm{tr}\nolimits} 

\DeclareMathOperator{\supp}{supp}

\title[Relativistic ground-state energy]{The leading behavior of the ground-state energy of heavy ions according to Brown and Ravenhall}

\author[Xiao Liu]{Xiao Liu}

\begin{document}

\begin{abstract}
In this article we prove the absence of relativistic effects in
leading order for the ground-state energy according to Brown-Ravenhall operator.
We obtain this asymptotic result for negative ions and for systems
with the number of electrons proportional to the nuclear charge. In
the case of neutral atoms the analogous result was obtained earlier
by Cassanas and Siedentop \cite{CassanasSiedentop2006}.
\end{abstract}

\maketitle

\tableofcontents
\pagebreak

\section{Introduction}
\label{sec1} Cassanas and Siedentop \cite{CassanasSiedentop2006}
have shown that the ground state energy of heavy atoms for the
relativistic Hamiltonian of Brown and Ravenhall is, to leading
order, given by the non-relativistic Thomas-Fermi energy. The
relativistic Hamiltonian of Brown and Ravenhall is derived from
quantum electrodynamics yielding energy levels correctly up to order
$\alpha^2$Ry \cite{BrownRavenhall1951}. However, only the case as
$N=Z$ is considered in \cite{CassanasSiedentop2006}, where $N$ is
the number of electrons and $Z$ is the nuclear charge. This thesis
describes the other two cases: $N>Z$ and $N/Z=\lambda (constant)$.

\section{Definition of The Model}
\label{sec2}

Brown and Ravenhall \cite{BrownRavenhall1951} describe two
relativistic electrons interacting with an external potential. The
model has an obvious generalization to the $N$-electron case.  First
we define
\begin{equation}
  \label{eq:d}
   D_{c,Z}:= \boldsymbol{\alpha}\cdot \frac c\ri\nabla + c^2\beta -Z
|\cdot|^{-1}
\end{equation}
 is the Dirac operator of an electron in the field of
a nucleus of charge $Z$. Note that we are using atomic units in this
paper, i.e., $m_e=\hbar=e=1$. As usual, the four matrices
$\alpha_1,\alpha_2,\alpha_3$ and $\beta$ are the four Dirac matrices
in standard representation, explicitly
$$\boldsymbol{\alpha}=\begin{pmatrix}
          0 & \boldsymbol\sigma \\
          \boldsymbol\sigma & 0
          \end{pmatrix},$$
$\boldsymbol\sigma$ denoting the three Pauli matrices
$$\sigma_1=\begin{pmatrix}
            0 & 1 \\
            1 & 0
            \end{pmatrix}, \qquad
  \sigma_2=\begin{pmatrix}
            0 & -\ri \\
            \ri & 0
            \end{pmatrix}, \qquad
  \sigma_3=\begin{pmatrix}
            1 & 0 \\
            0 & -1
            \end{pmatrix},$$
and $$\beta=\begin{pmatrix}
            1 & 0 & 0 & 0\\
            0 & 1 & 0 & 0\\
            0 & 0 & -1 & 0\\
            0 & 0 & 0 & -1
            \end{pmatrix}.$$
Let $\gQ_N:=\bigwedge_{\nu=1}^N(H^{1/2}(\rz^3)\otimes\cz^4)\cap
\gH_N$ where
\begin{equation}
  \label{eq:2}
  \gH_N:=\bigwedge_{\nu=1}^N\gH;
\end{equation}
the underlying one-particle Hilbert space is
\begin{equation}
  \label{eq:1}
  \gH := [\chi_{(0,\infty)}(D_{c,0})](L^2(\rz^3)\otimes\cz^4)
\end{equation}
where
\begin{equation}
  \label{eq:e}
   \chi_{(0,\infty)}(D_{c,0}):=
         \begin{cases}
            1,  & (\psi, D_{c,0}\psi)>0;\\
            0,  & (\psi, D_{c,0}\psi)\leq0
         \end{cases}
\end{equation} for any $\psi \in L^2(\rz^3)\otimes\cz^4$,
and
$$D_{c,0}:= \boldsymbol{\alpha}\cdot \frac c\ri\nabla + c^2\beta.$$

Now we define the model as
\begin{equation}
  \label{eq:f}
  \begin{split}
    \cE:\ &\gQ_N\rightarrow \rz\\
    &\psi \mapsto (\psi, (\sum_{\nu=1}^N(D_{c,0} -c^2-Z/|\bx_\nu|)_\nu+
      \sum_{1\leq\mu<\nu\leq N} |\bx_\mu-\bx_\nu|^{-1})\psi).
  \end{split}
\end{equation}

As an immediate consequence of the work of Evans et al.
\cite{Evansetal1996} quadratic form $\cE$ is bounded from below, in
fact it is positive (Tix \cite{Tix1997,Tix1998}), if
$\kappa:=Z/c\leq \kappa_\mathrm{crit}:=2/(\pi/2+2/\pi)$.  According
to Friedrichs this allows us to define a self-adjoint operator
$B_{c,N,Z}$ whose ground state energy
\begin{equation}
  \label{eq:5}
  E(c,N,Z):= \inf\sigma(B_{c,N,Z})=\inf\{\cE(\psi)|\psi\in\gQ_N,
  \|\psi\|=1\}
\end{equation}
is of concern to us in this paper. In fact the main result of this
article is
\begin{theorem}
  \label{t:haupt}
  When $N>Z$ or $N/Z=\lambda (constant)$, we have
  $$E(Z/\kappa,N,Z)=E_\mathrm{TF}(N,Z)+o(Z^{7/3})$$
  where $E_\mathrm{TF}(N,Z):=\inf\{\cE_\mathrm{TF}(\rho)\ |\rho\geq0,\ \int_{\rz^3}\rho\leq N,\ \rho \in
L^{5/3}(\rz^3)\}$, $$\cE_\mathrm{TF}(\rho):=
\int_{\rz^3}\left[\frac35\gamma_\mathrm{TF}\rho(\bx)^{5/3} - \frac
Z{|\bx|} \rho(\bx)\right]\rd \bx +
  D(\rho,\rho)$$ is the
Thomas-Fermi functional, and $\kappa:=Z/c\leq
\kappa_\mathrm{crit}:=2/\pi$.
\end{theorem}
In the following, we will assume that the ratio
$\kappa\in[0,\kappa_\mathrm{crit})$ is fixed. Note that according to
\cite[Formula (9c)]{LiebSimon1977} the minimal energy
$E_\mathrm{TF}(N,Z)$ fulfills the scaling relation
\begin{equation}
  \label{eq:15a-1}
  E_\mathrm{TF}(N,Z)=E_\mathrm{TF}(N/Z,1)Z^{7/3}.
\end{equation}

 The article is
structured as follows: first we show how the treatment of the
Brown-Ravenhall model can be reduced from Dirac spinors (4-spinors)
to Pauli spinors (2-spinors). Then we separately give the upper and
lower bounds for the case $N>Z$ in Section 3 and for the case where
$N/Z$ is fixed in Section 4. Throughout the paper, we use the letter
$k$ for any constants independent of $c$, $N$, $R$, or $Z$.

We now indicate, how to reduce to Pauli spinors. To this end we
parameterize the allowed states: Any $\psi\in \gH$ can be written as
\begin{equation}
  \label{eq:6}
  \displaystyle
  \psi :=
  \begin{pmatrix}
    {E_c(\p)+c^2\over N_c(\p)}u\\
    {c\p\cdot\boldsymbol\sigma\over N_c(\p)}u
  \end{pmatrix}
\end{equation}
for some $u\in \gh:=L^2(\rz^3)\otimes \cz^2$. Here,
$\boldsymbol{\sigma}$ are the three Pauli matrices,
$$\p:=-\ri\nabla,\ \ E_c(\p):= (c^2\p^2+c^4)^{1/2},\
N_c(\p):=[2E_c(\p)(E_c(\p)+c^2)]^{1/2}.$$  In fact, the map
\begin{equation}
  \label{eq:7}
  \begin{split}
    \Phi: \gh &\rightarrow \gH\\
  u&\mapsto
  \begin{pmatrix}
    \Phi_1u\\
    \Phi_2u
  \end{pmatrix}
:=
\begin{pmatrix}
    {E_c(\p)+c^2\over N_c(\p)}u\\
    {c\p\cdot\boldsymbol\sigma\over N_c(\p)}u
  \end{pmatrix}
\end{split}
\end{equation}
embeds $\gh$ unitarily into $\gH$ and its restriction onto
$H^1(\rz^3)\otimes\cz^2$ is also a unitary mapping to $\gH\cap
H^1(\rz^3)\otimes\cz^4$ (Evans et al. \cite{Evansetal1996}).

It suffices to study the energy as function of $u$
\begin{equation}
  \label{eq:8}
  \cE\circ (\otimes_{\nu=1}^N\Phi): \bigwedge_{\nu=1}^N \gh \rightarrow \rz.
\end{equation}

The one-particle Brown-Ravenhall operator $B_\gamma$ for an electron
the external electric potential of a point nucleus acting on Pauli
spinors is then (see Appendix \ref{sa3})
\begin{equation}
  \label{eq:9}
  B_{c,Z}= E_c(\p)- Z \varphi_1 -Z \varphi_2
\end{equation}
where
\begin{equation}
  \label{eq:9.5}
  B_{c,Z}:= \chi_{(0,\infty)}(D_{c,0})D_{c,Z}
\end{equation}
and we have split the potential into
\begin{equation}
  \label{eq:10}
  \varphi_1:=\Phi_1^*|\cdot|^{-1}\Phi_1,\ \
  \varphi_2:=\Phi_2^*|\cdot|^{-1}\Phi_2.
\end{equation}
As we will see the first part $\varphi_1$ is contributing to the
nonrelativistic limit whereas the second part turns out to give the
energy contribution that does not even affect the first correction
term.

\section{Case I: $N>Z$}
\label{sec3} In this section, we prove Theorem \ref{t:haupt} for
negatively charged atoms.

\subsection{Coherent States\label{ss3.1}}
We obtain the upper bound by constructing a trial density matrix in
the Hartree-Fock functional for the Brown-Ravenhall operator. To
this end we introduce spinor valued coherent states.

Given functions $f, \tilde f\in H^{3/2}(\rz^3)$ and an element
$\alpha=(\bp,\bq,\tau)$ of the phase space
$\Gamma:=\rz^3\times\rz^3\times\{1,2\}$, we define coherent states
in $\gh$ as
\begin{equation}
  \label{eq:11n}
   F_\alpha(x):= (\varphi_{\bp,\bq}\otimes e_\tau)(x)
  :=f(\bx-\bq)\exp(i\bp\cdot\bx)\delta_{\tau,\sigma},
\end{equation}
where $x=(\bx,\sigma)\in\rz^3\times\{1,2\}$ and the vectors $e_\tau$
are the canonical basis vectors in $\cz^2$ (see Lieb \cite{Lieb1981}
and Evans et al. \cite{Evansetal1996}). We also define
\begin{equation}
  \label{eq:11n-1}
  \phi_{\tilde{k}}(\bx):=(2^{\tilde{k}})^{-3/2}\tilde f(\tfrac{\bx}{2^{\tilde k}}),
\end{equation} where $\tilde{k}\in\nz$.
We will pick $f$ depending on a dilation parameter. More
specifically, we will choose
\begin{equation}
  \label{eq:11an}
   f(\bx):=g_R(\bx):=R^{-3/2}g(R^{-1}\bx)
\end{equation}
(see \cite[Formula (11)]{CassanasSiedentop2006}),
\begin{equation}
  \label{eq:11an-1}
   \tilde f(\bx):=
         \begin{cases}
            (2 \pi \tilde R)^{-1/2} \mid |\bx|-\tilde R \mid ^{-1} \sin ({\pi \mid |\bx|-\tilde R \mid}/\tilde R),  & \tilde R \leq |\bx| \leq 2\tilde R;\\
               0,      & otherwise,
         \end{cases}
\end{equation}
where $R:=Z^{-\delta}$ with $\delta\in(1/3,2/3)$, $\tilde R \in
\rz_+$ and $g\in H^{3/2}$, spherically symmetric, normalized, and
with support in the unit ball.

The natural measure on $\Gamma$ counting the number of electrons per
phase space volume in the spirit of Planck is $\int_\Gamma
\dbar\Omega(\alpha) := (2\pi)^{-3}\int \rd \bp \int \rd\bq
\sum_{\tau=1}^2$. The essential properties needed are the following.
For $A, \tilde A\in L^1(\Gamma,\dbar \Omega)$, let
$$\tilde{\gamma_1}:=\int_\Gamma\dbar \Omega(\alpha)
  \tilde A(\alpha)|F_\alpha\rangle\langle F_\alpha|,$$
$$\gamma_2:=\sum\limits_{ \tilde{k}=K+1}^{K+N-Z} |\phi_{
\tilde{k}}\rangle\langle \phi_{\tilde{k}}|,$$
$$\tilde{\gamma_2}:=\epsilon_{\tilde R}
\sum\limits_{\tilde{k}=K+1}^{K+Z}
|\phi_{N-Z+\tilde{k}}\rangle\langle \phi_{N-Z+\tilde{k}}|, \quad
where \ \  \epsilon_{\tilde R}:=1-\frac{\int_\Gamma\dbar
\Omega(\alpha)
  \tilde A(\alpha)}{Z}.$$
Then
$$\gamma:=\tilde{\gamma_1}+\gamma_2+\tilde{\gamma_2}$$
and
\begin{equation}
  \label{eq:11bn}
  \gamma_1:= \int_\Gamma\dbar \Omega(\alpha)
  A(\alpha)|F_\alpha\rangle\langle F_\alpha|
\end{equation}
(see \cite[Formula (12)]{CassanasSiedentop2006}) are trace class
operators.

We will pick
\begin{equation}
  \label{eq:14an}
  A(\alpha):=\chi_{\{(\bp,\bq)\in\rz^6|\bp^2/2 -
  V_Z(\bq)\leq 0\}}(\bp,\bq)
\end{equation}
(see \cite[Formula (16)]{CassanasSiedentop2006}),
\begin{equation}
  \label{eq:14an2}
  \tilde A(\alpha):=\chi_{\{(\bp,\bq)\in\rz^6|\bp^2/2 -
  V_Z(\bq)\leq 0,\ |\bq|\leq \tilde R - R\}}(\bp,\bq),
\end{equation}
where $V_Z:=Z/|\cdot|-|\cdot|^{-1}*\rho^{(N,Z)}_\mathrm{TF}$; here
$\rho^{(N,Z)}_\mathrm{TF}$ is the unique minimizer of the
Thomas-Fermi functional
\begin{equation}
  \label{eq:15-1}
  \cE_\mathrm{TF}(\rho):= \int_{\rz^3}\left[\frac35\gamma_\mathrm{TF}\rho(\bx)^{5/3} - \frac Z{|\bx|} \rho(\bx)\right]\rd \bx +
  D(\rho,\rho)
\end{equation}
in the set of functions $\rho \in \{\rho\ |\rho\geq0,\
\int_{\rz^3}\rho\leq N,\ \rho \in L^{5/3}(\rz^3)\}$. Here
$$D(\rho,\rho):=\frac12\int_{\rz^3}\int_{\rz^3}\rho(\bx)|\bx-\by|^{-1}\rho(\by)\rd\bx\rd\by$$
is the Coulomb scalar product. For fermions with $q$ spin states per
particle, $\gamma_\mathrm{TF}:=(6\pi^2/q)^{2/3}\hbar^2/(2m)$(see
Lieb \cite[Formula (2.6)]{Lieb1981}), i.e., in our units,
$\gamma_\mathrm{TF}=(3\pi^2)^{2/3}/2$.

According to \cite{Loss2005},
$E_\mathrm{TF}(N,Z):=\inf\cE_\mathrm{TF}(\rho)$ is strictly monotone
decreasing for $N\leq Z$ and constant for $N>Z$. Thus for $N>Z$, the
minimizer $\rho$ of the Thomas-Fermi functional \eqref{eq:15-1}
coincides the minimizer in the case $N=Z$, which we denote by
$\rho_\mathrm{Z}$. Hence
$V_Z=Z/|\cdot|-|\cdot|^{-1}*\rho_\mathrm{Z}$ for $N>Z$.

  By Lemma \ref{a1},
\begin{equation}
  \label{eq:12-1}
  0\leq A\leq1 \implies 0\leq \gamma_1\leq 1,
\end{equation}
according to the definition of $\tilde A$ and the support sets of
$\phi_k$,
\begin{equation}
  \label{eq:12n}
  0\leq \tilde A\leq1 \implies 0\leq \gamma\leq 1;
\end{equation}
by Lemma \ref{a2}
\begin{equation}
\label{eq:13-1} \tr\gamma_1 = \int_\Gamma \dbar \Omega(\alpha)
A(\alpha) \leq N,
\end{equation}
\begin{equation}
\label{eq:13-1n} \tr\gamma =
\tr\tilde{\gamma_1}+\tr\gamma_2+\tr\tilde{\gamma_2} =N.
\end{equation}
Using $\Phi$ we can lift $\gamma$ to an operator on $\gH$
\begin{equation}
  \label{eq:14-1}
  \gamma_\Phi := \Phi\gamma\Phi^*.
\end{equation}
Note that
\begin{equation}
 \label{eq:13-2n}
 V^{(N)}_Z(\bq):= Z^{4/3}V^{(N/Z)}_1(Z^{1/3}\bq)
\end{equation}
(see Gomb\'as \cite{Gombas1949} and \cite{LiebSimon1977}). Note also
that \eqref{eq:13-2n} does not depend on $N$ for $N>Z$.

\subsection{Upper Bound on ${\cE}^R_\mathrm{HF}(\gamma)$\label{ss3.2}}

We begin by noting that the Hartree-Fock functional -- with or
without exchange energy -- bounds $E(c,N,Z)$ from above. To be exact
we introduce the set of density matrices
\begin{equation}
  \label{eq:16n}
  S_{\partial N}:=\{\gamma\in \gS^1(\gh)\ | E_c(\hat\bp)\gamma\in \gS^1(\gh),\ 0\leq\gamma\leq1,\ \tr\gamma=N \}
\end{equation}
where $\gS^1(\gh)$ denotes the trace class operators on $\gh$.
\begin{equation}
  \label{eq:17-1}
  \begin{split}
    {\cE}^R_\mathrm{HF}: S_{\partial N}&\rightarrow \rz\\
    \gamma&\mapsto \tr[(E_c(\hat\bp)-c^2 - Z/|\bx|)\gamma_\Phi]+ D(\rho_{\gamma_\Phi},\rho_{\gamma_\Phi})
  \end{split}
\end{equation}
where -- as usual -- $\rho_\gamma(\bx):=\gamma(\bx,\bx)$ is the
density associated to $\gamma$ and $D$ is the Coulomb scalar
product. By the analogue of Lieb's result \cite[Formula
(9)]{Lieb1981V} and \cite{Lieb1981E} (see also Bach \cite{Bach1992})
-- which trivially transcribes from the Schr\"odinger setting to the
present one -- we have for all $\gamma\in S_{\partial N}$
\begin{equation}
  \label{eq:18-1}
  E(c,N,Z)\leq {\cE}^R_\mathrm{HF}(\gamma).
\end{equation}

\subsubsection{Kinetic Energy of $\tilde{\gamma_1}$\label{sss3.2.1}}

\begin{lemma}
  \label{3.2.1}
  The kinetic energy of $\tilde{\gamma_1}$ does not exceed
  the kinetic energy of $\gamma_1$, i.e.,
  \begin{equation}
    \label{eq:3.2.1-1}
    \tr[(E_c(\hat\bp)-c^2)\tilde{\gamma_1}_\Phi]\leq
    \tr[(E_c(\hat\bp)-c^2){\gamma_1}_\Phi].
  \end{equation}
\end{lemma}
\begin{proof}
According to
\begin{equation}
  \label{eq:19-1}
  E_c(\p)-c^2\geq 0,
\end{equation}
we have
  \begin{multline}
   \label{eq:3.2.1}
   \tr[(E_c(\hat\bp)-c^2)(\tilde{\gamma_1}_\Phi-{\gamma_1}_\Phi)]=\int \int_\Gamma\dbar
\Omega(\alpha)
  [\tilde A(\alpha)-A(\alpha)][(E_c(\hat\bp)-c^2)F_\alpha(x)] \overline {F_\alpha(x)}\rd x
  \\
  = -\frac12 \int (2\pi)^{-3} \int\limits_{{\bp}^2/2 - V_Z(\bq)\leq 0 \atop |\bq|> \tilde R - R} \sum_{\tau=1}^2 \rd\bp \rd\bq [(E_c(\hat\bp)-c^2)^{1/2}F_\alpha(x)] \overline {[(E_c(\hat\bp)-c^2)^{1/2}F_\alpha(x)]}\rd
  x \\
    =-(2\pi)^{-3} \int\limits_{{\bp}^2/2 - V_Z(\bq)\leq 0 \atop |\bq|> \tilde R - R} \int |(E_c(\hat\bp)-c^2)^{1/2}F_\alpha(x)|^2 \rd x \rd\bp \rd\bq \leq
    0. \\
  \end{multline}
  Thus $\tr[(E_c(\hat\bp)-c^2)\tilde{\gamma_1}_\Phi]\leq \tr[(E_c(\hat\bp)-c^2){\gamma_1}_\Phi]$.
\end{proof}

\subsubsection{External Potential of $\tilde{\gamma_1}$\label{sss3.2.2}}

\begin{lemma}
  \label{3.2.2}
  For any $\varepsilon >0$, there exists $\tilde{R}$
  large enough such that\\ $\tr[(- Z/|\bx|)\tilde{\gamma_1}_\Phi]\leq \tr[(- Z/|\bx|){\gamma_1}_\Phi]+\varepsilon$
, obviously for any fixed $\tilde{R}$, we have\\ $\tr[(-
Z/|\bx|)\tilde{\gamma_1}_\Phi]\geq \tr[(- Z/|\bx|){\gamma_1}_\Phi]$,
i.e.,
  \begin{equation}
    \label{eq:3.2.2-1}
    \tr[(- Z/|\bx|)\tilde{\gamma_1}_\Phi]\longrightarrow
    \tr[(- Z/|\bx|){\gamma_1}_\Phi], \ \ \tilde R \rightarrow
    \infty.
  \end{equation}
\end{lemma}

\begin{proof}
According to Lieb \cite[Formula (2.18)]{Lieb1981}, we can write the
Thomas-Fermi equation as
\begin{equation}
  \label{eq:a00}
  \gamma_\mathrm{TF}\rho_\mathrm{Z}^{2/3}(\bx)=[V_Z(\bx)-u']_+,
\end{equation}
where $u'\geq0$ is some constant and for $t\in\rz$, we set
$[t]_+:=\max\{ t,0 \}$. Then for $N>Z$, the Thomas-Fermi potential
$V_Z:=Z/|\cdot|-\rho_\mathrm{Z}*|\cdot|^{-1}$ can be written as
\begin{equation}
  \label{eq:a0}
  \gamma_\mathrm{TF}\rho_\mathrm{Z}^{2/3}=V_Z
\end{equation}
(see, e.g., Gomb\'as \cite{Gombas1949}).

 Using Kato's inequality (see
\cite[Formula (2.9)]{Herbst1977}), passing from $\bxi$ to
$\tfrac{\bxi}R+\bp$ and taking into account that $$ |\hat
g(\bxi)|^2[|\bxi|^2 +
  2|\bxi| +1]$$ is integrable
because $g\in H^{3/2}(\rz^3)$, we can get
  \begin{multline}
   \label{eq:3.2.2}
    \tr[(- Z/|\bx|)(\tilde{\gamma_1}_\Phi-{\gamma_1}_\Phi)]\leq
    \frac\pi2 Z \int \int_\Gamma\dbar
\Omega(\alpha)
  [A(\alpha)-\tilde A(\alpha)]|\hat F_\alpha(\bxi)|^2 |\bxi|\rd
  \bxi \\ =\frac\pi2 Z \int (2\pi)^{-3} \int\limits_{{\bp}^2/2 - V_Z(\bq)\leq 0 \atop |\bq|> \tilde R - R} \sum_{\tau=1}^2 \rd\bp \rd\bq |\hat F_\alpha(\bxi)|^2 |\bxi|\rd \bxi \\
  \leq k Z \int_\Gamma\dbar \Omega(\alpha)[A(\alpha)-\tilde A(\alpha)]R^{-4}\int |\hat g(\bxi)|^2|\bxi+R\bp|
  \rd\bxi \\
  \leq k Z \int_\Gamma\dbar \Omega(\alpha)[A(\alpha)-\tilde
  A(\alpha)]R^{-4}(1+R|\bp|) \\
  \leq k Z^{1+4\delta}\int\limits_{{\bp}^2/2 - \gamma_\mathrm{TF}\rho_\mathrm{Z}^{2/3}(\bq)\leq 0 \atop |\bq|> \tilde R - R}\rd\bp
  \rd\bq (1+R|\bp|) \\
  \leq k Z^{1+4\delta}\int\limits_{|\bq|> \tilde R -
  R}[\rho_\mathrm{Z}(\bq)+R\rho_\mathrm{Z}^{4/3}(\bq)]\rd\bq\longrightarrow 0, \quad when \quad \tilde R \rightarrow\infty. \\
  \end{multline}
  The last step is according to absolute continuity of Lebesgue
  integral.\\
  Thus $\tr[(- Z/|\bx|)\tilde{\gamma_1}_\Phi]\longrightarrow \tr[(-
  Z/|\bx|){\gamma_1}_\Phi]$,  when $\tilde R \rightarrow \infty$.
\end{proof}

\subsubsection{The Electron-Electron Interaction of $\tilde{\gamma_1}$\label{sss3.2.3}}

\begin{lemma}
  \label{3.2.3}
  The electron-electron interaction of $\tilde{\gamma_1}$ does not
  exceed the electron-electron interaction of $\gamma_1$, i.e.,
  \begin{equation}
    \label{eq:3.2.3-1}
    D(\rho_{\tilde{\gamma_1}_\Phi},\rho_{\tilde{\gamma_1}_\Phi})\leq D(\rho_{{\gamma_1}_\Phi},\rho_{{\gamma_1}_\Phi}).
  \end{equation}
\end{lemma}
\begin{proof}
  \begin{multline}
   \label{eq:3.2.3}
   D(\rho_{\tilde{\gamma_1}_\Phi},\rho_{\tilde{\gamma_1}_\Phi})\\=\frac12 \iint \int_\Gamma\dbar
\Omega(\alpha)
  \tilde A(\alpha)|\Phi F_\alpha(x)|^2 |x-y|^{-1} \int_\Gamma\dbar
\Omega(\alpha)
  \tilde A(\alpha)|\Phi F_\alpha(y)|^2 \rd x \rd y \\ \leq \frac12 \iint \int_\Gamma\dbar
\Omega(\alpha)
   A(\alpha)|\Phi F_\alpha(x)|^2 |x-y|^{-1} \int_\Gamma\dbar
\Omega(\alpha)
   A(\alpha)|\Phi F_\alpha(y)|^2 \rd x \rd y\\=D(\rho_{{\gamma_1}_\Phi},\rho_{{\gamma_1}_\Phi}).
  \end{multline}
\end{proof}

\subsubsection{Kinetic Energy of $\gamma_2$\label{sss3.2.4}}

We introduce the set of density matrices
\begin{equation}
  \label{eq:16}
  S_N:=\{\gamma_1\in \gS^1(\gh)\ | E_c(\hat\bp)\gamma_1\in \gS^1(\gh),\ 0\leq\gamma_1\leq1,\ \tr\gamma_1\leq N \}
\end{equation}
where $\gS^1(\gh)$ denotes the trace class operators on $\gh$.
\begin{lemma}
  \label{3.2.4}
  For any $\varepsilon >0$, there exists a large enough $K$
  such that\\ $\tr[(E_c(\hat\bp)-c^2){\gamma_2}_\Phi]\leq \varepsilon$
, i.e.,
  \begin{equation}
    \label{eq:3.2.4-1}
    \tr[(E_c(\hat\bp)-c^2){\gamma_2}_\Phi]\longrightarrow
    0, \ \ K \rightarrow \infty.
  \end{equation}
\end{lemma}
\begin{proof}
By concavity we have
\begin{equation}
  \label{eq:19}
  E_c(\p)-c^2\leq \tfrac12\p^2=-\tfrac12 \Delta,
\end{equation}
which implies that the Brown-Ravenhall kinetic energy is bounded by
the non-relativistic one, i.e., for all $\gamma_2 \in S_N$ with
$-\Delta\gamma_2\in \gS^1(\gh)$
\begin{equation}
  \label{eq:20-3.2.4}
  \tr[(E_c(\hat\bp)-c^2)\gamma_2]\leq \tr(-\tfrac12\Delta\gamma_2).
\end{equation}
Then
  \begin{multline}
   \label{eq:3.2.4}
    \tr[(E_c(\hat\bp)-c^2){\gamma_2}_\Phi]\leq \tr(-\tfrac12
    \Delta\gamma_2) =\frac12 \sum\limits_{\tilde{k}=K+1}^{K+N-Z} (2^{\tilde{k}})^{-2} \|\nabla \tilde f\|^2 \\
    \leq \frac12 \sum\limits_{\tilde{k}=K+1}^{\infty} 4^{-\tilde{k}} \|\nabla \tilde f\|^2
    = \frac12 {{(\frac14)^{K+1}}\over{1-\frac14}} \|\nabla \tilde f\|^2 \\
    =\frac23 (\frac14)^{K+1} \|\nabla \tilde f\|^2 \longrightarrow 0, \quad when \quad K \rightarrow \infty. \\
  \end{multline}
    Thus $\tr[(E_c(\hat\bp)-c^2){\gamma_2}_\Phi]\longrightarrow 0$,  when $K \rightarrow \infty$.
\end{proof}

\subsubsection{External Potential of $\gamma_2$\label{sss3.2.5}}

\begin{lemma}
  \label{3.2.5}
  The external potential of $\gamma_2$ does not exceed zero, i.e.,
  \begin{equation}
    \label{eq:3.2.5-1}
    \tr[(- Z/|\bx|){\gamma_2}_\Phi]\leq 0.
  \end{equation}
\end{lemma}
\begin{proof}
  \begin{multline}
   \label{eq:3.2.5}
    \tr[(Z/|\bx|){\gamma_2}_\Phi]=Z \sum_{\tilde{k}=K+1}^{K+N-Z}(\Phi\phi_{\tilde{k}}, {1\over{|\bx|}}\Phi\phi_{\tilde{k}}) =Z \sum_{\tilde{k}=K+1}^{K+N-Z}\int|\Phi\phi_{\tilde{k}}(\bx)|^2 {1\over{|\bx|}}\rd\bx\geq
    0.
  \end{multline}
\end{proof}

\subsubsection{The Electron-Electron Interaction of $\gamma_2$\label{sss3.2.6}}

\begin{lemma}
  \label{3.2.6}
  For any $\varepsilon >0$, there exists a large enough $K$
  such that\\ $D(\rho_{{\gamma_2}_\Phi},\rho_{{\gamma_2}_\Phi})\leq \varepsilon$
, i.e.,
  \begin{equation}
    \label{eq:3.2.6-1}
    D(\rho_{{\gamma_2}_\Phi},\rho_{{\gamma_2}_\Phi})\longrightarrow
    0, \ \ K \rightarrow \infty.
  \end{equation}
\end{lemma}
\begin{proof}
According to \cite[Equation (12)]{Barbaroux2004}, we can get
  \begin{multline}
   \label{eq:3.2.6}
   D(\rho_{{\gamma_2}_\Phi},\rho_{{\gamma_2}_\Phi})\leq \frac\pi4 \tr |{\gamma_2}_\Phi| \tr (\sqrt{-\Delta} |{\gamma_2}_\Phi|) \\
   \leq \frac\pi4(N-Z) \sum_{\tilde{k}=K+1}^{K+N-Z}(\Phi\phi_{\tilde{k}}, |\p|\Phi\phi_{\tilde{k}}) \\
   \leq \frac\pi4(N-Z) \sum_{\tilde{k}=K+1}^{K+N-Z}\int|\hat \phi_{\tilde{k}}(\p)|^2 |\p| \rd\p \\
   \leq \sum_{\tilde{k}=K+1}^{\infty} \frac\pi4(N-Z)\int (2^{\tilde{k}})^{-1} |\hat {\tilde f}(\bxi)|^2 |\bxi| \rd\bxi \\
   = \frac\pi4 (N-Z) (\frac12)^K \int |\hat {\tilde f}(\bxi)|^2 |\bxi|\rd
   \bxi \longrightarrow 0, \quad when \quad K \rightarrow \infty. \\
  \end{multline}
  Thus $D(\rho_{{\gamma_2}_\Phi},\rho_{{\gamma_2}_\Phi})\longrightarrow 0$,  when $K \rightarrow \infty$.
\end{proof}

\subsubsection{The Total Energy of ${\cE}^R_\mathrm{HF}(\gamma)$\label{sss3.2.7}}

We define the reduced Hartree-Fock functional on $\gamma_1$ as
following
\begin{equation}
  \label{eq:17}
  \begin{split}
    {\cE}^R_\mathrm{HF}: S_N&\rightarrow \rz\\
    \gamma_1&\mapsto \tr[(E_c(\hat\bp)-c^2 - Z/|\bx|){\gamma_1}_\Phi]+ D(\rho_{{\gamma_1}_\Phi},\rho_{{\gamma_1}_\Phi})
  \end{split}
\end{equation}
where -- as usual -- $\rho_{\gamma_1}$ is the density associated to
$\gamma_1$ and $D$ is the Coulomb scalar product. Gathering our
above estimates allows us to get Theorem \ref{3.2.7}.
\begin{theorem}
  \label{3.2.7}
  The reduced Hartree-Fock functional of $\gamma$ does not exceed the reduced Hartree-Fock functional of $\gamma_1$, i.e.,
  \begin{equation}
    \label{eq:1-3.2.7}
    {\cE}^R_\mathrm{HF}(\gamma)\leq {\cE}^R_\mathrm{HF}(\gamma_1),
  \end{equation}
  when $\tilde R$ and $K$ are tending to infinity.
\end{theorem}
\begin{proof}
According to the definition of $\tilde{\gamma_2}$ and Appendix
\ref{sa2}, we get that as $\tilde R$ tends to infinity,
$\epsilon_{\tilde R}$ tends to zero, and thus
${\cE}^R_\mathrm{HF}(\tilde{\gamma_2})$ tends to zero. Using the
results obtained in Lemmata \ref{3.2.1} through \ref{3.2.6}, we get
for any $\varepsilon >0$, there exist large enough $\tilde R$ and
$K$
  such that
  \begin{multline}
    \label{eq:3.2.7}
{\cE}^R_\mathrm{HF}(\gamma)=\tr[(E_c(\hat\bp)-c^2 -
Z/|\bx|)\gamma_\Phi]+ D(\rho_{\gamma_\Phi},\rho_{\gamma_\Phi})
\\ \leq \tr[(E_c(\hat\bp)-c^2 - Z/|\bx|){\gamma_1}_\Phi]+
D(\rho_{{\gamma_1}_\Phi},\rho_{{\gamma_1}_\Phi})+\varepsilon={\cE}^R_\mathrm{HF}(\gamma_1)+\varepsilon.
  \end{multline}
\end{proof}

\subsection{Upper Bound\label{ss3.3}}

We begin by noting that the Hartree-Fock functional -- with or
without exchange energy -- bounds $E(c,N,Z)$ from above. To be exact
we introduce the set of density matrices \eqref{eq:16}, where
$\gS^1(\gh)$ denotes the trace class operators on $\gh$. We define
the reduced Hartree-Fock functional of $\gamma_1$ as \eqref{eq:17}.
By Theorem \ref{3.2.7} and the analogon of Lieb's result
\cite{Lieb1981V,Lieb1981E} (see also Bach \cite{Bach1992}) -- which
trivially transcribes from the Schr\"odinger setting to the present
one -- we have for all $\gamma_1\in S_N$
\begin{equation}
  \label{eq:18}
  E(c,N,Z)\leq {\cE}^R_\mathrm{HF}(\gamma)\leq {\cE}^R_\mathrm{HF}(\gamma_1).
\end{equation}

\subsubsection{Kinetic Energy\label{sss3.3.1}}

\eqref{eq:19} implies that the Brown-Ravenhall kinetic energy is
bounded by the non-relativistic one, i.e., for all $\gamma_1 \in
S_N$ with $-\Delta\gamma_1\in \gS^1(\gh)$
\begin{equation}
  \label{eq:20}
  \tr[(E_c(\hat\bp)-c^2)\gamma_1]\leq \tr(-\tfrac12\Delta\gamma_1).
\end{equation}
Insertion of $\gamma_1$ (see Equations \eqref{eq:11n},
\eqref{eq:11an}, \eqref{eq:11bn}, and \eqref{eq:14an}) turns the
right hand side into the Thomas-Fermi kinetic energy modulo the
positive error $(\tr\gamma_1) \|\nabla g\|^2 R^{-2}$ (see Lieb
\cite[Formula (5.9)]{Lieb1981}).

In fact, we choose $$f(\bx)=(2\pi R)^{-1/2}|\bx|^{-1} {\sin
(\pi|\bx|/R)}$$ (see Lieb \cite[Formula (5.11)]{Lieb1981}). Because
of $f(\bx)=R^{-3/2}g(R^{-1}\bx)$, we know
$$g(\bx)=(2\pi)^{-1/2}|\bx|^{-1} \sin (\pi|\bx|).$$ Let
$\boldsymbol{\eta}=R^{-1}\bx,$ we can calculate
$$\nabla f=R^{-3/2}\nabla g(R^{-1}\bx)=R^{-3/2}\nabla_{\boldsymbol{\eta}} {g(\boldsymbol{\eta})}{{\partial\boldsymbol{\eta}}\over{\partial\bx}}=R^{-3/2}\nabla_{\boldsymbol{\eta}} {g(\boldsymbol{\eta})} {1\over R} {{\bx}\over{|\bx|}},$$
$$|\nabla f|=R^{-5/2}|\nabla_{\boldsymbol{\eta}} {g(\boldsymbol{\eta})}|,\  \int |\nabla f|^2\rd^3\bx=\int R^{-5}|\nabla_{\boldsymbol{\eta}} {g(\boldsymbol{\eta})}|^2(R^3\rd^3\boldsymbol{\eta})=R^{-2}\|\nabla g\|^2.$$
Thus we obtain
\begin{equation}
  \label{eq:21}
  \tr[(E_c(\hat\bp)-c^2)\gamma_1]
  \leq \frac35\gamma_\mathrm{TF}\int \rho_\mathrm{Z}^{5/3}(\bx)\rd\bx
  + ZR^{-2}\|\nabla g\|^2.
\end{equation}

\subsubsection{External Potential\label{sss3.3.2}}

Since $-Z\tr(\varphi_2\gamma_1)$ is negative, we can and will
estimate this term by zero. This estimate will be good, if this term
is of smaller order. Although, logically unnecessary for the upper
bound, it is interesting to see that $\varphi_2$ does indeed not
significantly contribute to the energy, if $\gamma_1$ is chosen as
above. Moreover, the proof will be also useful for the proof of
Lemma \ref{l2}.
\begin{lemma}
  \label{l1}
  For our choice of
  $\gamma_1=\int_\Gamma\dbar\Omega(\alpha)A(\alpha)|F_\alpha\rangle\langle
  F_\alpha|$ and $\delta\in(1/3,2/3)$ we have
  \begin{equation}
    \label{eq:22}
    0\leq Z\tr(\varphi_2\gamma_1)
    \leq k Z\int_\Gamma\dbar\Omega(\alpha)A(\alpha)
    \iint\rd \bxi \rd \bxi'{c^2|\bxi||\bxi'||\hat F_\alpha(\bxi)||\hat F_\alpha(\bxi')| \over |\bxi-\bxi'|^2 N_c(\bxi)N_c(\bxi')}
    = O(Z^{4/3+\delta}).
  \end{equation}
\end{lemma}

\begin{proof}
  We begin by estimating the expectation of $\varphi_2$ on the coherent
  state \eqref{eq:11n}.
  \begin{multline}
    \label{eq:23}
    0 \leq (F_\alpha, \varphi_2 F_\alpha) \leq k \iint \rd \bxi \rd
    \bxi' {c^2|\bxi||\bxi'| |\hat F_\alpha(\bxi)| |\hat F_\alpha(\bxi')| \over
      N_c(\bxi)|\bxi-\bxi'|^2
      N_c(\bxi')} \\
    \leq k c^{-2}R^{-3} \iint \rd\bxi\rd\bxi'{|\hat g(\bxi)||\hat
      g(\bxi')|\over |\bxi-\bxi'|^2}|\bxi+R\bp||\bxi'+R\bp| \leq k
    c^{-2}R^{-3} (1+ R|\bp|+R^2|\bp|^2).
  \end{multline}
  Here, we have used that $N_c(\bxi)\geq\sqrt2c^2$ and, at the last step,
  that
  $$
  {|\hat g(\bxi)||\hat g(\bxi')|\over |\bxi-\bxi'|^2}[|\bxi||\bxi'| +
  |\bxi|+|\bxi'| +1]
  $$
  is integrable in $\bxi$ and $\bxi'$ because $g\in H^{3/2}(\rz^3)$. Thus according to \eqref{eq:23}, we get
  \begin{multline}
    \label{eq:24}
    0\leq Z\tr(\varphi_2\gamma_1)=Z\int \dbar\Omega(\alpha) A(\alpha)(F_\alpha, \varphi_2 F_\alpha)\\
    \leq k{Z\over c^2R^3} \int \dbar\Omega(\alpha) A(\alpha)(1+ R|\bp|+R^2|\bp|^2)\\
    \leq  k{Z\over c^2R^3}\left\{Z+ R\int\rd \bq \left[Z^{4/3} V_1(Z^{1/3}\bq)\right]^2+R^2\int\rd \bq \left[Z^{4/3} V_1(Z^{1/3}\bq)\right]^{5/2}\right\}\\
    = O(Z^{3\delta} + Z^{2/3+2\delta} + Z^{4/3+\delta})
  \end{multline}
  (see \cite[Formula (27)]{CassanasSiedentop2006}).
\end{proof}
\begin{lemma}
  \label{l2}
  For our choice of $\gamma_1$ and $\delta\in(1/3,2/3)$ we have
  \begin{multline}
    \label{eq:25}
    \left|Z\tr[(|\cdot|^{-1}-\varphi_1)\gamma_1]\right| \\
    \leq k Z\int\dbar\Omega(\alpha)A(\alpha) \iint\frac{\rd\boldsymbol{\xi}\rd\boldsymbol{\xi}'}{|\boldsymbol{\xi}-\boldsymbol{\xi}'|^2}
    \left(1-{(E_c(\boldsymbol{\xi})+c^2)(E_c(\boldsymbol{\xi}')+c^2) \over
      N_c(\boldsymbol{\xi})N_c(\boldsymbol{\xi}')}\right) |\hat F_\alpha(\boldsymbol{\xi})||\hat F_\alpha(\boldsymbol{\xi}')| \\
    = O(Z^{5/3+\delta}).
  \end{multline}
\end{lemma}
\begin{proof}
  We first note that
  \begin{multline}
    \label{eq:26a}
    \left|1-{(E_c(\boldsymbol{\xi})+c^2)(E_c(\boldsymbol{\xi}')+c^2)
        \over N_c(\boldsymbol{\xi})N_c(\boldsymbol{\xi}')}\right| \\
    \leq
    {\left|3E_c(\boldsymbol{\xi})E_c(\boldsymbol{\xi}')-c^2(E_c(\boldsymbol{\xi})+E_c(\boldsymbol{\xi}')+c^2)\right|
      \over N_c(\boldsymbol{\xi})N_c(\boldsymbol{\xi}')}.
  \end{multline}
  Then, noting that $E_c(\boldsymbol{\xi})-c^2\leq
  c|\boldsymbol{\xi}|$, we obtain
  \begin{multline}
  \label{eq26b}
  \left|1-{(E_c(\boldsymbol{\xi})+c^2)(E_c(\boldsymbol{\xi}')+c^2) \over
      N_c(\boldsymbol{\xi})N_c(\boldsymbol{\xi}')}\right|
  \leq {3c^2|\boldsymbol{\xi}||\boldsymbol{\xi}'|+2c^3(|\boldsymbol{\xi}|+|\boldsymbol{\xi}'|)\over
    N_c(\boldsymbol{\xi})N_c(\boldsymbol{\xi}')}\\
  \leq{3c^2|\boldsymbol{\xi}||\boldsymbol{\xi}'|+2c^3(|\boldsymbol{\xi}|+|\boldsymbol{\xi}'|)\over
      2c^4} .
  \end{multline}
  Using this last equation, we estimate
  \begin{multline}
    \label{eq:26}
    |(F_\alpha, (\frac1{|\cdot|}-\varphi_1) F_\alpha)|
    \\
    \leq k
    \iint\frac{\rd\boldsymbol{\xi}\rd\boldsymbol{\xi}'}{|\boldsymbol{\xi}-\boldsymbol{\xi}'|^2}
    \left(1-{(E_c(\boldsymbol{\xi})+c^2)(E_c(\boldsymbol{\xi}')+c^2)
        \over
        N_c(\boldsymbol{\xi})N_c(\boldsymbol{\xi}')}\right) |\hat F_\alpha(\boldsymbol{\xi})||\hat F_\alpha(\boldsymbol{\xi}')| \\
    \leq k \int_{\rz^6} \rd \boldsymbol{\xi}\rd\boldsymbol{\xi}'
    {|\hat g_R(\boldsymbol{\xi}-\bp)| |\hat g_R(\boldsymbol{\xi}'-\bp)|
      \over |\boldsymbol{\xi}-\boldsymbol{\xi}'|^2}   (c^{-2}|\boldsymbol{\xi}||\boldsymbol{\xi}'| +c^{-1} (|\boldsymbol{\xi}|+|\boldsymbol{\xi}'|))\\
    \leq k c^{-2}R^{-3} \int
    \rd\boldsymbol{\xi}\int\rd\boldsymbol{\xi}'{|\hat
      g(\boldsymbol{\xi})||\hat
      g(\boldsymbol{\xi}')|\over |\boldsymbol{\xi}-\boldsymbol{\xi}'|^2}( |\boldsymbol{\xi}+R\bp||\boldsymbol{\xi}'+R\bp| + c R|\boldsymbol{\xi}+R\bp|)\\
    \leq k c^{-2}R^{-3} \int
    \rd\boldsymbol{\xi}\int\rd\boldsymbol{\xi}'{|\hat
      g(\boldsymbol{\xi})||\hat g(\boldsymbol{\xi}')|\over
      |\boldsymbol{\xi}-\boldsymbol{\xi}'|^2}(
    |\boldsymbol{\xi}||\boldsymbol{\xi}'|+R|\bp||\boldsymbol{\xi}|+|R\bp|^2
    + cR|\boldsymbol{\xi}|+cR^2|\bp|)\\
    \leq k c^{-2}R^{-3} (1+ R|\bp|+R^2|\bp|^2+cR +cR^2|\bp|)
  \end{multline}
  (see \cite[Formula (31)]{CassanasSiedentop2006}).
  Thus
  \begin{multline}
    \label{eq:27}
    Z|\tr[(|\cdot|^{-1}-\varphi_1)\gamma_1]\leq
    Z|\int_\Gamma\dbar\Omega(\alpha)
    A(\alpha)(F_\alpha, (|\cdot|^{-1}-\varphi_1) F_\alpha)|\\
    \leq k Z\int\dbar\Omega(\alpha)A(\alpha) c^{-2}R^{-3} (1+ R|\bp|+R^2|\bp|^2+cR
    +cR^2|\bp|)\\
    \leq  k{Z\over c^2R^3}\bigg\{Z+ R\int\rd \bq \left[Z^{4/3} V_1(Z^{1/3}\bq)\right]^2+R^2\int\rd \bq \left[Z^{4/3}
    V_1(Z^{1/3}\bq)\right]^{5/2}\\
    +cR+cR^2\int\rd \bq \left[Z^{4/3} V_1(Z^{1/3}\bq)\right]^2\bigg\}\\
    \leq k(Z^{3\delta}+ Z^{2\delta + 2/3}+ Z^{\delta+4/3}+
    Z^{2\delta}+ Z^{\delta+5/3})
  \end{multline}
  which yields the desired estimate.
\end{proof}

\subsubsection{The Electron-Electron Interaction\label{sss3.3.3}}

We will roll back the treatment of the electron-electron interaction
to the treatment of nucleus-electron interaction.

\begin{lemma}
  \label{l3}
  For our choice of $\gamma_1$ and $\delta\in(1/3,2/3)$ we have
  \begin{equation}
    \label{eq:28}
    D(\rho_{{\gamma_1}_\Phi},\rho_{{\gamma_1}_\Phi})-D(\rho_{\gamma_1},\rho_{\gamma_1})= O(Z^{5/3+\delta}),
  \end{equation}
  where $\rho_{\gamma_1}$ is the density of $\gamma_1$ and $\rho_{{\gamma_1}_\Phi}$ is the
  density of ${\gamma_1}_\Phi$.
\end{lemma}
\begin{proof}
  We have
  \begin{equation}
    \label{eq:29}
    |\cF[(\rho_{\gamma_1}+\rho_{{\gamma_1}_\Phi})*|\cdot|^{-1}](\bxi)|\leq 2^{-3/2}\pi^{-5/2}\|\rho_{\gamma_1} + \rho_{{\gamma_1}_\Phi}\|_1|\bxi|^{-2} =2^{-1/2}\pi^{-5/2} Z|\bxi|^{-2}.
  \end{equation}
  Now (see \cite[Lemma 3]{CassanasSiedentop2006}),
\begin{multline*}
  |D(\rho_{{\gamma_1}_\Phi},\rho_{{\gamma_1}_\Phi})-D(\rho_{\gamma_1},\rho_{\gamma_1})|=
  |D(\rho_{{\gamma_1}_\Phi}-\rho_{\gamma_1},\rho_{{\gamma_1}_\Phi}+\rho_{\gamma_1})|\\
  \leq \frac12\left|\int_{\rz^3}(\rho_{\gamma_1}(\bx)-\rho_{{\gamma_1}_\Phi}(\bx))[(\rho_{\gamma_1}+\rho_{{\gamma_1}_\Phi})*|\cdot|^{-1}](\bx)\rd \bx\right|\\
  \leq \frac12\int_\Gamma \dbar\Omega(\alpha)A(\alpha)\\
  \times\iint \rd\bxi \rd\bxi'|
  \cF[(\rho_{\gamma_1}+\rho_{{\gamma_1}_\Phi})*|\cdot|^{-1}](\bxi-\bxi')|K(\bxi,\bxi')
  |\hat
  F_\alpha(\bxi)||\hat F_\alpha(\bxi')| \rd\bxi\rd\bxi'\\
  \leq 2^{-3/2}\pi^{-5/2} Z\int_\Gamma \dbar\Omega(\alpha)A(\alpha)\iint \rd\bxi
  \rd\bxi'| |\bxi-\bxi'|^{-2} K(\bxi,\bxi') |\hat F_\alpha(\bxi)||\hat
  F_\alpha(\bxi')| \rd\bxi\rd\bxi'
\end{multline*}
where
$$K(\bxi,\bxi')=\left|\frac{(E_c(\bxi)+c^2)(E_c(\bxi')+c^2)}{N_c(\bxi)N_c(\bxi')}-1 \right| +  \frac{c^2|\bxi||\bxi'|}{N_c(\bxi)N_c(\bxi')}$$
and where we used \eqref{eq:29} in the last step. Eventually,
Lemmata \ref{l1} and \ref{l2} yield the desired result.
\end{proof}

\subsubsection{The Total Energy\label{sss3.3.4}}

Gathering our above estimates allows us to reduce the problem to the
non-relativistic result of Lieb \cite{Lieb1981}
\begin{theorem}
    \label{th:3}
  There exist a constant $k$ such that for all $Z\geq 1$ we have $$E(Z/\kappa , N,Z)\leq E_\mathrm{TF}(N/Z,1)Z^{7/3} + k Z^{20/9}.$$
\end{theorem}
\begin{proof}
  Following Lieb \cite[Section V.A.1]{Lieb1981} with the remainder
  terms given there (putting $R=Z^{-\delta}$ as in our estimate),
  using the remainder terms obtained in Lemmata \ref{l1} through
  \ref{l3}, and applying \eqref{eq:3.2.7} and \eqref{eq:21} we get
  \begin{equation}
    \label{eq:31}
    E(c,N,Z)\leq {\cE}^R_\mathrm{HF}(\gamma)\leq {\cE}^R_\mathrm{HF}(\gamma_1) \leq E_\mathrm{TF}(N,Z) +
    O(Z^{1+2\delta} + Z^{\frac52-\frac\delta2} + Z^{\frac53+\delta})
  \end{equation}
  (see \cite[Formula (35)]{CassanasSiedentop2006}) which is optimized for $\delta=5/9$ giving the claimed result.
\end{proof}

\subsection{Lower Bound\label{ss3.4}}

The lower bound is -- contrary to the usual folklore -- easy. As we
will see, it is a corollary of S\o rensen's \cite{Sorensen2005}
result for the Chandrasekhar operator and an estimate on the
potential generated by the exchange hole \cite{Mancasetal2004}.  The
exchange hole of a density $\sigma$ at a point $\bx\in\rz^3$ is
defined as the ball $B_{R_\sigma(\bx)}(\bx)$ of radius
$R_\sigma(\bx)$ centered at $\bx$ where $R_\sigma(\bx)$ is the
smallest radius $R$ fulfilling
  \begin{equation}
    \label{eq:31a}
    \frac12=\int_{B_R} \sigma.
  \end{equation}
The hole potential $L_\sigma $ of $\sigma$ is defined through
\begin{equation}
  \label{eq:31b}
  L_\sigma(\bx):=  \int_{|\bx-\by|<R_\sigma(\bx)}{\sigma(\by) \over |\bx-\by| }\rd\by.
\end{equation}

\subsubsection{$L^\infty$-Bound on the Exchange Hole Potential\label{sss3.4.1}}

We begin with the following remark: the Thomas-Fermi potential
$V_Z:=Z/|\cdot|-\rho_\mathrm{Z}*|\cdot|^{-1}$ can be written as
\eqref{eq:a0}. This equation yields immediately the upper bound
\begin{equation}
  \label{eq:a1}
  \rho_\mathrm{Z}(\bx)\leq (Z/\gamma_\mathrm{TF})^{3/2}|\bx|^{-3/2}.
\end{equation}
This bound allows us to prove the following $L^\infty$-bounds on
potentials of exchange holes.
\begin{lemma}
  \label{l4}
  $$\|L_{\rho_\mathrm{Z}}\|_\infty = O(Z).$$
\end{lemma}
\begin{proof}
  The function
  \begin{equation}
    \label{eq:a3}
    \begin{split}
      f:\rz_+&\rightarrow \rz\\
      t&\mapsto \sqrt t\int_{|\by|<1/t}|\by|^{-1}|\by+(0,0,1)|^{-3/2}\rd \by
    \end{split}
  \end{equation}
  is obviously continuous on $(0,\infty)$. Moreover, $f(t)$ tends to a
  positive constant for $t\to0$ and to $0$ for $t\to\infty$. Thus,
  $\|f\|_\infty<\infty$.

  This allows us to obtain the desired estimate:
  \begin{equation}
    \label{eq:a4}
    L_{\rho_\mathrm{Z}}(\bx) \leq A_1(\bx) + A_2(\bx)
  \end{equation}
  (see \cite[Formula (46)]{CassanasSiedentop2006}) where
    \begin{multline}
      \label{eq:a5}
      A_1(\bx):=\int_{|y|\leq 1/Z} \frac{\rho_\mathrm{Z}(\bx+\by)}{|\by|} \rd
      \by
      \leq \left(\frac Z{\gamma_\mathrm{TF}}\right)^{3/2}
      \int_{|y|\leq 1/Z} \frac{\rd \by}{|\by||\by+\bx|^{3/2}} \\
      = (Z/\gamma_\mathrm{TF})^{3/2} Z^{-1/2} f(|\bx|Z) \leq
      \|f\|_\infty\gamma_\mathrm{TF}^{-3/2} Z.
\end{multline}
and
\begin{equation}
  \label{eq:a6}
  A_2(\bx):=\int_{\frac1Z\leq |\by| \leq R_{\rho_\mathrm{Z}}(\bx)}
  \frac{\rho_\mathrm{Z} (\bx+\by)}{|\by|} \rd \by \leq Z
  \int_{\frac1Z\leq |\by| \leq R_{\rho_\mathrm{Z}}(\bx)}
  \rho_\mathrm{Z}(\bx+\by)\rd \by
  \leq \frac Z2.
\end{equation}
These two estimates prove the claim.
\end{proof}

Lemma \ref{l4} allows us to estimate the $N$ electron operator
$B_{c,N,Z}$ by the canonical one particle Brown-Ravenhall operator
whose nuclear charge is screened by the the Thomas-Fermi potential.
However, since we would like -- because of mere convenience -- to
take advantage of S\o rensen's result \cite{Sorensen2005}, we derive
an estimate on $L_{\rho_{\delta}}$ (where
$\rho_\delta:=\rho_\mathrm{Z}* g^2_{Z^{-\delta}}$), i.e., the
exchange hole potential of the density occurring in S\o rensen's
proof.
\begin{lemma}
\label{l8}
$$\| L_{\rho_\delta}\|_\infty=O(Z).$$
\end{lemma}
\begin{proof}
  We proceed analogously to the proof of Lemma \ref{l4}:
  \begin{multline}
    \label{eq:a7}
    L_{\rho_\delta}(\bx) \leq \int_{|\by|\leq1/Z}
    \frac{\rho_\delta(\bx+\by)}{|\by|} \rd \by + \int_{1/Z\leq |\by|
      \leq R_{\rho_{\delta}}(\bx)}
    \frac{\rho_\delta (\bx+\by)}{|\by|} \rd \by\\
    \leq \int \rd \bz g^2_{Z^{-\delta}}(\bz)\int_{|\by|\leq1/Z}
    \frac{\rho_{\mathrm{Z}}(\bx-\bz+\by)}{|\by|} \rd \by +
    Z\int_{|\by|\leq R_{\rho_{\delta}}(\bx)}
    \rho_\delta(\bx+\by) \rd \by \\
    \leq \int \rd \bz g^2_{Z^{-\delta}}(\bz) A_1(\bx-\bz) + \frac Z2\leq k Z
  \end{multline}
  (see \cite[Formula (49)]{CassanasSiedentop2006}) where we used the definition of the radius of the exchange hole
  from the
  second line to the third line, the definition of $A_1$ in the next
  step, and in the last step the $L^\infty$-estimate \eqref{eq:a5} on
  $A_1$.
\end{proof}

\subsubsection{Lower Bound\label{sss3.4.2}}

\begin{theorem}
  \label{t2}
  $$\liminf_{Z\to\infty}[E(c,N,Z)-E_\mathrm{TF}(N,Z)]Z^{-7/3} \geq 0.$$
\end{theorem}
\begin{proof}
  Pick $\delta>0$ and set $\rho_\delta:=\rho_\mathrm{Z}*g_{Z^{-\delta}}^2.$ Then
  the exchange hole correlation bound \cite[Equation (14)]{Mancasetal2004} implies
  the following pointwise estimate
  \begin{equation}
    \label{eq:32}
    \sum_{1\leq\mu<\nu\leq N}{1\over|\bx_\mu-\bx_\nu|}\geq
    \sum_{\nu=1}^N[\rho_\delta*|\cdot|^{-1}(\bx_\nu)- L_{\rho_\delta}(\bx_\nu)]-
    D(\rho_\delta,\rho_\delta).
  \end{equation}
  Because of the spherical symmetry of $g$ we can use Newton's theorem
  \cite{Newton1972} and replace $\rho_\delta$ by $\rho_\mathrm{Z}$ in
  the third summand of the right hand side of \eqref{eq:32}. Then, by
  Lemma \ref{l8}, we get that for all normalized $\psi\in\gQ_N$
  \begin{equation}
    \label{eq:32a}
    \cE(\psi)    \geq \tr[\Lambda_+(|D_0|-c^2-V_\delta)\Lambda_+]_-
    -kNZ -D(\rho_{\mathrm{Z}},\rho_{\mathrm{Z}})
  \end{equation}
  where, for $t\in\rz$, we set $[t]_-:=\min\{ t,0 \}$ and
  $V_\delta=Z/|\cdot|-\rho_\delta*|\cdot|^{-1}$.

  To count the number of spin states per electron correctly, i.e., two
  instead of the apparent four, we use an observation by Lieb et al.
  \cite[Appendix B]{Liebetal1997}: Note that
  \begin{equation}
    \label{eq:32b}
    \Lambda_-=U^{-1}\Lambda_+ \, U,\qquad \mbox{ where }\quad
    U:=\left(
      \begin{array}{cc}
        0    & 1\\
        -1 & 0
      \end{array}\right).
  \end{equation}
  Indeed, we have
  $$\Lambda_-=\frac12\left(1-\frac{D_0}{|D_0|}\right), \qquad
  \Lambda_+=\frac12\left(1+\frac{D_0}{|D_0|}\right)$$
  and
  $$UD_0 \, U^{-1}=
  \left(\begin{array}{cc}
      0    & 1\\
      -1 & 0
    \end{array}\right)
  \left(\begin{array}{cc}
      mc^2         & c\,\boldsymbol\sigma.\p\\
      c\, \boldsymbol\sigma.\hat{\bp} & -mc^2
    \end{array}\right)
  \left(\begin{array}{cc}
      0   & -1\\
      1 & 0
    \end{array}\right)=-D_0.$$
  We set $X:=(|D_0|-c^2-V_\delta(\bx))I_2$, and write
  $$\tr\left[\Lambda_+\left(\begin{array}{cc}
        X   & 0\\
        0   & X
      \end{array}\right)\Lambda_+\right]_-\geq \tr\left( \Lambda_+
    \left(\begin{array}{cc}
        X_-   & 0\\
        0     & X_-
      \end{array}\right)\Lambda_+\right)=
  \tr\left( \Lambda_+
    \left(\begin{array}{cc}
        X_-   & 0\\
        0     & X_-
      \end{array}\right)\right)$$
  $$\tr\left( \Lambda_-
    \left(\begin{array}{cc}
        X_-   & 0\\
        0     & X_-
      \end{array}\right)\right)=
  \tr\left( \Lambda_+ U
    \left(\begin{array}{cc}
        X_-   & 0\\
        0     & X_-
      \end{array}\right)U\right)=
  \tr\left( \Lambda_+
    \left(\begin{array}{cc}
        X_-   & 0\\
        0     & X_-
      \end{array}\right)\right)$$
  Thus
  \begin{multline}
    \label{eq:32c}
    2\tr\left( \Lambda_+ \left(\begin{array}{cc}
          X_-   & 0\\
          0 & X_-
        \end{array}\right)\right)\\
    = \tr\left( \Lambda_+ \left(\begin{array}{cc}
          X_-   & 0\\
          0 & X_-
        \end{array}\right)\right)+
    \tr\left( \Lambda_-
      \left(\begin{array}{cc}
          X_-   & 0\\
          0     & X_-
        \end{array}\right)\right)
    =2\tr(X_-).
  \end{multline}
  Since $|D_0|=E_c(\p)$ and $X$ is a 2 by 2 matrix, we obtain
  \begin{multline}
   \tr\left[\Lambda_+\left(|D_0|-c^2-V_\delta(\bx)\right)\Lambda_+\right]_-=
   \tr\left[\Lambda_+\left(\begin{array}{cc}
        X   & 0\\
        0   & X
      \end{array}\right)\Lambda_+\right]_-\\
   \geq\tr\left( \Lambda_+
    \left(\begin{array}{cc}
        X_-   & 0\\
        0     & X_-
      \end{array}\right)\right)=\tr(X_-)=\,2\,
      \tr[E_c(\p)-c^2-V_\delta(\bx)]_-.
  \end{multline}
Then
  \begin{equation}
    \label{eq:32d}
    E(Z/\kappa ,N,Z)\, \geq \,2\, \tr[E_c(\p)-c^2-V_\delta(\bx)]_-
    -D(\rho_{\mathrm{Z}},\rho_{\mathrm{Z}})-kNZ
  \end{equation}
  (see \cite[Formula (42)]{CassanasSiedentop2006}) where the last trace is spinless. This connects to S\o rensen's
  Equation (3.2) from \cite{Sorensen2005}. It is a fundamental result of \cite{Sigal1982} that $E(Z/\kappa ,N,Z)=E(Z/\kappa ,N_c(Z),Z)$, for any $N \geq N_c(Z)$, where $N_c$ is the number of negative particles that can be bound to an atom of nuclear charge $Z$. Considering $N_c<2Z+1$ (see \cite[Formula (1.2)]{Lieb1984}), when $Z \rightarrow \infty$, we can get
  \begin{equation}
    \label{eq:32e}
    E(Z/\kappa ,N,Z)\, \geq \,2\, \tr[E_c(\p)-c^2-V_\delta(\bx)]_-
    -D(\rho_{\mathrm{Z}},\rho_{\mathrm{Z}})-k\  O(Z)\  Z.
  \end{equation}
  This result then follows using his lower bound.
\end{proof}

\section{Case II: $N/Z=\lambda$ (Constant)}
\label{sec4}

The case $\lambda>1$ has been already solved in Section 3;
$\lambda=1$ is solved in \cite{CassanasSiedentop2006}, so it only
remains to consider $\lambda<1$.

\subsection{Coherent States\label{ss4.1}}

This Section is analogous to \ref{ss3.1}.

\subsection{Upper Bound\label{ss4.2}}

This Division is similar to \ref {ss3.2} and \ref {ss3.3}.
Analogously to Formula \eqref {eq:16n}, we introduce the set of
density matrices
\begin{equation}
  \label{seq:16}
  S_{\partial N}:=\{\gamma\in \gS^1(\gh)\ | E_c(\hat\bp)\gamma\in \gS^1(\gh),\ 0\leq\gamma\leq1,\ \tr\gamma=N=\lambda Z,\ \}
\end{equation}
where $\gS^1(\gh)$ denotes the trace class operators on $\gh$;
$\lambda$ is a number independent of $c$, $N$, $R$, or $Z$. Defining
${\cE}^R_\mathrm{HF}$ as Formula \eqref {eq:17-1}, and as the same
as above, for all $\gamma\in S_{\partial N}$, we have Formula \eqref
{eq:18-1}.

\subsubsection{Kinetic Energy\label{sss4.2.1}}

According to \cite{Loss2005}, we can write
$\rho^{(N,Z)}_\mathrm{TF}$ simply as $\rho_\mathrm{TF}$ in the
following. By concavity, and then analogously to Formula \eqref
{eq:20}, we know
\begin{equation}
  \label{seq:20}
  \tr[(E_c(\hat\bp)-c^2)\gamma]\leq \tr(-\tfrac12\Delta\gamma).
\end{equation}
Similarly to Formula \eqref {eq:21}, we get
\begin{equation}
  \label{seq:21}
  \tr[(E_c(\hat\bp)-c^2)\gamma]
  \leq \frac35\gamma_\mathrm{TF}\int \rho_\mathrm{TF}^{5/3}(\bx)\rd\bx
  + (\tr\gamma)R^{-2}\|\nabla g\|^2.
\end{equation}

\subsubsection{External Potential\label{sss4.2.2}}

As the same as Part \ref {sss3.3.2}, we can obtain following two
lemmata by the analogues of Lemmata \ref {l1} and \ref {l2}.
\begin{lemma}
  \label{l6}
  For our choice of
  $\gamma=\int_\Gamma\dbar\Omega(\alpha)A(\alpha)|F_\alpha\rangle\langle
  F_\alpha|$ and $\delta\in(1/3,2/3)$ we have
  \begin{equation}
    \label{seq:22}
    0\leq Z\tr(\varphi_2\gamma)
    \leq k Z\int_\Gamma\dbar\Omega(\alpha)A(\alpha)
    \iint\rd \bxi \rd \bxi'{c^2|\bxi||\bxi'||\hat F_\alpha(\bxi)||\hat F_\alpha(\bxi')| \over |\bxi-\bxi'|^2 N_c(\bxi)N_c(\bxi')}
    = O(Z^{4/3+\delta}).
  \end{equation}
\end{lemma}

\begin{lemma}
  \label{l7}
  For our choice of $\gamma$ and $\delta\in(1/3,2/3)$ we have
  \begin{multline}
    \label{seq:25}
    \left|Z\tr[(|\cdot|^{-1}-\varphi_1)\gamma]\right| \\
    \leq k Z\int\dbar\Omega(\alpha)A(\alpha) \iint\frac{\rd\boldsymbol{\xi}\rd\boldsymbol{\xi}'}{|\boldsymbol{\xi}-\boldsymbol{\xi}'|^2}
    \left(1-{(E_c(\boldsymbol{\xi})+c^2)(E_c(\boldsymbol{\xi}')+c^2) \over
      N_c(\boldsymbol{\xi})N_c(\boldsymbol{\xi}')}\right) |\hat F_\alpha(\boldsymbol{\xi})||\hat F_\alpha(\boldsymbol{\xi}')| \\
    = O(Z^{5/3+\delta}).
  \end{multline}
\end{lemma}

\subsubsection{The Electron-Electron Interaction\label{sss4.2.3}}

Similarly to Lemma \ref {l3}, we get Lemma \ref {sl8} below.

\begin{lemma}
  \label{sl8}
  For our choice of $\gamma$ and $\delta\in(1/3,2/3)$ we have
  \begin{equation}
    \label{seq:28}
    D(\rho_{\gamma_\Phi},\rho_{\gamma_\Phi})-D(\rho_\gamma,\rho_\gamma)= O(Z^{5/3+\delta}),
  \end{equation}
  where $\rho_\gamma$ is the density of $\gamma$ and $\rho_{\gamma_\Phi}$ is the
  density of $\gamma_\Phi$.
\end{lemma}

\subsubsection{The Total Energy\label{sss4.2.4}}

By the analogue of Theorem \ref {th:3}, we obtain Theorem \ref
{th:5}.
\begin{theorem}
    \label{th:5}
  There exist a constant $k$ such that we have for all $Z\geq 1$
  $$E(Z/\kappa , N,Z)\leq E_\mathrm{TF}(\lambda ,1)Z^{7/3} + k Z^{20/9}.$$
\end{theorem}

\subsection{Lower Bound\label{ss4.3}}

This Part is similar to \ref {ss3.4}.

\subsubsection{$L^\infty$-Bound on the Exchange Hole Potential\label{sss4.3.1}}

\begin{lemma}
  \label{sl4}
  $$\|L_{\rho_\mathrm{TF}}\|_\infty = O(Z).$$
\end{lemma}

Similarly to Lemma \ref {l8}, since we would like to take advantage
of S\o rensen's result \cite{Sorensen2005}, we derive an estimate on
$L_{\rho_{\delta}}$ (where $\rho_\delta:=\rho_\mathrm{TF}*
g^2_{Z^{-\delta}}$).
\begin{lemma}
\label{tl8}
$$\| L_{\rho_\delta}\|_\infty=O(Z).$$
\end{lemma}

\subsubsection{Lower Bound\label{sss4.3.2}}

\begin{theorem}
  \label{st2}
  $$\liminf_{Z\to\infty}[E(c,N,Z)-E_\mathrm{TF}(N,Z)]Z^{-7/3} \geq 0.$$
\end{theorem}

\appendix
\section{The Proof of $B_{c,Z}= E_c(\p)- Z \varphi_1 -Z \varphi_2$\label{sa3}}

\begin{lemma}
  \label{a3}
  The one-particle Brown-Ravenhall operator $B_\gamma$ for an electron
the external electric potential of a point nucleus acting on Pauli
spinors is
  $$B_{c,Z}= E_c(\p)- Z \varphi_1 -Z \varphi_2.$$
\end{lemma}

\begin{proof}
For any $\psi\in \gH$, we can get
  \begin{multline}
   \label{eq:an3}
   \cE(\psi)=(\psi,B_{c,Z}\psi)=(\psi,D_{c,Z}\psi)\\
   =\begin{pmatrix}
   {E_c(\p)+c^2\over N_c(\p)}u & {c\p\cdot\boldsymbol\sigma\over N_c(\p)}u
   \end{pmatrix}
   \left(\begin{array}{cc}
      c^2-Z|\cdot|^{-1}        & c\,\boldsymbol\sigma\cdot\p\\
      c\, \boldsymbol\sigma\cdot\p & -c^2-Z|\cdot|^{-1}
   \end{array}\right)
   \begin{pmatrix}
    {E_c(\p)+c^2\over N_c(\p)}u\\
    {c\p\cdot\boldsymbol\sigma\over N_c(\p)}u
   \end{pmatrix}\\
   =\begin{pmatrix}
   \Phi_1u & \Phi_2u
   \end{pmatrix}
   \begin{pmatrix}
    (c^2-Z|\cdot|^{-1})\,\Phi_1u+c\,\boldsymbol\sigma\cdot\p\,\Phi_2u\\
    c\,\boldsymbol\sigma\cdot\p\,\Phi_1u+(-c^2-Z|\cdot|^{-1})\,\Phi_2u
   \end{pmatrix}\\
   =\left(u,\left(\Phi_1c^2\Phi_1+\Phi_1c\,\boldsymbol\sigma\cdot\p\,\Phi_2+\Phi_2c\,\boldsymbol\sigma\cdot\p\,\Phi_1-\Phi_2c^2\Phi_2\right)u\right)-Z(u,\varphi_1u)-Z(u,\varphi_2u)\\
   =(u,E_c(\p)u)-Z(u,\varphi_1u)-Z(u,\varphi_2u).
  \end{multline}
\end{proof}

\section{Checking of $\gamma_1$ and $\gamma$\label{sa1}}

\begin{lemma}
  \label{a1}
  $$0\leq \gamma_1\leq 1 \ \ \  and \ \ \ 0\leq \gamma\leq 1.$$
\end{lemma}

\begin{proof}
For any normalized $u\in \gh$, using Parseval's equality, we get
  \begin{multline}
   \label{eq:an1}
   (u, \gamma_1u)\\=(2\pi)^{-3}\int u(\bx)\iint \overline{f(\bx-\bq)}f(\bx'-\bq)A(\alpha)\exp(-i\bp\cdot\bx)\exp(i\bp\cdot\bx')\overline{u(\bx')}\rd\bp\rd\bq\rd\bx'\rd\bx\\
   =\int(2\pi)^{-3/2}\int\exp(-i\bp\cdot\bx)u(\bx)\overline{f(\bx-\bq)}\rd\bx A(\alpha)\\
   \times \overline{(2\pi)^{-3/2}\int\exp(-i\bp\cdot\bx')u(\bx')\overline{f(\bx'-\bq)}\rd\bx'}\rd\bp\rd\bq\\
   =\int\left|(2\pi)^{-3/2}\int\exp(-i\bp\cdot\bx)u(\bx)\overline{f(\bx-\bq)}\rd\bx\right|^2A(\alpha)\rd\bp\rd\bq\\
   \leq \int\left|(2\pi)^{-3/2}\int\exp(-i\bp\cdot\bx)u(\bx)\overline{f(\bx-\bq)}\rd\bx\right|^2\rd\bp\rd\bq \\
   =\iint|u(\bx)|^2|f(\bx-\bq)|^2\rd\bx\rd\bq\\
   =\int|u(\bx)|^2\left[\int|f(\bx-\bq)|^2\rd\bq\right]\rd\bx=\|u\|^2=1,\\
  \end{multline}
  and following \eqref{eq:an1},
  \begin{multline}
   \label{eq:an1-1}
   (u, \gamma u)=(u, \tilde{\gamma_1} u)+(u, \gamma_2 u)+(u, \tilde{\gamma_2} u)\\
   \leq \int\limits_{|\bx|\leq \tilde R - R}|u(\bx)|^2\rd\bx+\sum\limits_{\tilde{k}=K+1}^{K+N-Z}\int\limits_{\bx\in \supp\{\phi_{ \tilde{k}}\}}|u(\bx)|^2\rd\bx\\
   +\sum\limits_{\tilde{k}=K+N-Z+1}^{K+N}\int\limits_{\ \bx\in \supp\{\phi_{ \tilde{k}}\}}|u(\bx)|^2\rd\bx \\
   \leq\|u\|^2=1.\\
  \end{multline}
\end{proof}

\section{Checking of $\tr\gamma_1$ and $\tr\gamma$\label{sa2}}

\begin{lemma}
  \label{a2}
  $$\tr\gamma_1\leq N \ \ \  and \ \ \ \tr\gamma=N.$$
\end{lemma}

\begin{proof}
Using formula \eqref{eq:a0}, we can get
  \begin{multline}
   \label{eq:an2}
   \tr\gamma_1 = \int_\Gamma \dbar \Omega(\alpha)
A(\alpha)=(2\pi)^{-3} \int\limits_{{\bp}^2/2 - V_Z(\bq)\leq 0} \sum_{\tau=1}^2 \rd\bp \rd\bq\\
     =2(2\pi)^{-3}\frac{4\pi}3(2\gamma_\mathrm{TF})^{3/2}\int \rho_\mathrm{Z}(\bq)\rd\bq=Z\leq N, \\
  \end{multline}
  and
  \begin{multline}
   \label{eq:an2-2}
   \tr\gamma = \int_\Gamma \dbar \Omega(\alpha)\tilde
A(\alpha) + (N-Z) + \epsilon_{\tilde R} Z =N. \\
  \end{multline}
\end{proof}

\paragraph{Acknowledgment}: I am very grateful to Professor Doctor Heinz Siedentop
for his constant attention to this work. I would like to give my
most appreciation for his encouragement and helping. I am also
deeply grateful for patience and opportune support of Roch Cassanas,
Marco Maceda, Oliver Matte, Edgardo Stockmeyer, Sergey Morozov and
Matthias Huber.

\end{document}